\documentclass[11pt]{article}

\usepackage[margin=1in]{geometry}

\usepackage{times}

\usepackage{soul}
\usepackage{url}
\usepackage[hidelinks]{hyperref}
\usepackage[utf8]{inputenc}
\usepackage[small]{caption}
\usepackage{graphicx,amsthm}
\usepackage{amsmath}
\usepackage{amsfonts}
\usepackage{booktabs}
\usepackage{color}
\usepackage{xfrac}
\usepackage{paralist}
\usepackage[ruled,linesnumbered,vlined]{algorithm2e}
\usepackage{authblk}

\newtheorem{claim}{Claim}[section]
\newtheorem{theorem}{Theorem}[section]

\newtheorem{lemma}[theorem]{Lemma}

\newtheorem{definition}[theorem]{Definition}
\newtheorem{invariant}{Invariant}[section]

\newcommand{\bfs}{\mathbf{s}} 
\newcommand{\bfv}{\mathbf{v}} 
\newcommand{\bfB}{\mathbf{B}}

\title{Approximately Envy-Free Budget-Feasible Allocation\thanks{The authors are ordered alphabetically. The authors thank Edith Elkind, Georgios Birmpas, Warut Suksompong, and Alexandros Voudouris for helpful discussions at the early stage of this work.}}

\author[1]{Jiarui Gan}
\author[2]{Bo Li}
\author[3]{Xiaowei Wu}
\affil[1]{Max Planck Institute for Software Systems}
\affil[2]{Department of Computing, The Hong Kong Polytechnic University}
\affil[3]{IOTSC, University of Macau}
\affil[ ]{\texttt{jrgan@mpi-sws.org, comp-bo.li@polyu.edu.hk, xiaoweiwu@um.edu.mo}}

\date{}

\begin{document}

\maketitle

\begin{abstract}
	In the budget-feasible allocation problem, a set of items with varied sizes and values are to be allocated to a group of agents. Each agent has a budget constraint on the total size of items she can receive.
	The goal is to compute a feasible allocation that is \emph{envy-free} (EF), in which the agents do not envy each other for the items they receive, nor do they envy a charity, who is endowed with all the unallocated items.
	Since EF allocations barely exist even without budget constraints, we are interested in the relaxed notion of {\em envy-freeness up to one item} (EF1).	
% --- Older version ---
%	In the recent work by Wu et al. (IJCAI 2021), it is shown that the budget-feasible allocation maximizing Nash Social Welfare (NSW) achieves an approximation of 1/4 for EF1.
%	However, since the computation of such allocations does not admit polynomial time algorithms, the computation of (approximately) EF1 allocation remains largely open.
%%	
%	In this paper, we take one step towards solving this problem by showing that for agents with identical additive valuations, a 1/2-approximate EF1 allocation can be computed in polynomial time.
%	We also propose efficient algorithms to compute an exact EF1 allocation for the uniform-budget case and the two-agent case.
%	Finally, we show that for identical additive valuations the approximation ratio of EF1 offered by an NSW maximizing allocation approaches 1 when the item sizes are infinitesimal compared with the budgets.
	The computation of both exact and approximate EF1 allocations remains largely open, despite a recent effort by Wu et al. (IJCAI 2021) in showing that any budget-feasible allocation that maximizes the Nash Social Welfare (NSW) is 1/4-approximate EF1.
	In this paper, we move one step forward by showing that for agents with identical additive valuations, a 1/2-approximate EF1 allocation can be computed in polynomial time.
	For the uniform-budget and two-agent cases, we  propose efficient algorithms for computing an exact EF1 allocation.
	We also consider the large budget setting, i.e., when the item sizes are infinitesimal compared with the agents' budgets, and show that both the NSW maximizing allocation and the allocation our polynomial-time algorithm computes have an approximation close to 1 regarding EF1.
\end{abstract}

\section{Introduction}

In the fair allocation problem of indivisible items, a set of items $M$ with different values are to be allocated to a group of $n$ agents $N$.
The goal is to compute an allocation that is \emph{envy-free} (EF), in which no agent finds that the bundle of items of any other agent is more valuable than her own.
Previous research mostly considers \emph{unconstrained} versions of this problem, in which any $n$-partitioning of the item set $M$ is accepted as a feasible allocation.
However, in many real world applications, there might be budget constraints that prevents an agent from taking any arbitrary bundle of items.
Consider a scenario where a company, i.e., a contractor, outsources a number of projects (items) to some subcontractors (agents). Each project comes with a profit (value) and a workload (size); a subcontractor invests the required workload to accomplish each project assigned to them and earns a profit.
Of course, each subcontractor wants to earn as much profit as they can, as long as the total workload is budget-feasible (i.e., it does not exceed their budget).
Meanwhile, to maintain long-term partnerships with the subcontractors, the company does not want any of them to feel less prioritized in project assignments, so fairness is a primary concern.

A natural model for such applications is recently proposed by Wu et al.~\cite{corr/WuLG20}.
%To this end, it is natural to consider the following budget-feasible setting of the fair allocation problem that has recently been proposed in~\cite{corr/WuLG20}.
In their model, each item $j\in M$ has a size $s_j$, and each agent $i\in N$ has a budget $B_i > 0$, which puts an upper bound on the total size of items this agent can receive.
An associated EF notion is also proposed.
By this EF notion, an agent $i$ does not envy an agent $j$ if and only if she cannot find a {\em subset} of agent $j$'s items that is more valuable than her own, while the size of this subset does not exceed her budget.
Indeed, it would be unfair if we ignore the agents' budgets and require an agent to receive as much value as another agent, even though her budget is much larger that of the other agent. 

%In the budget-feasible setting, in addition to the value $v_j$, each item $j\in M$ is also associated with a size $s_j$; each agent $i\in N$ has budget $B_i > 0$, which upper limits the total size of items agent $i$ receives.
%%
%Under the budget constraints, we say that agent $i$ does not envy $j$ if she cannot find a {\em subset} of agent $j$'s bundle which is more valuable than her own while the size of this bundle does not exceed her budget.
%%
%This EF notion incorporates budget constraints; indeed, it would be unfair if we ignore the agents' budgets and require an agent to receive as much value as another agent, even though her budget is much larger that of the other agent. 

The budget constraints may very often prevent the items from being all allocated (e.g., in one extreme where every agent's budget is even smaller than the smallest item). Hence, unlike many other fair allocation tasks, having all items allocated to the agents is not required in our setting. Nevertheless, without this requirement, a trivial EF solution would then be allocating nothing to every agent, which completely defeats our purpose of allocating the items.
For a more meaningful objective, a dummy agent called the {\em charity} is introduced, who is assumed to have a budget large enough for taking all the items, but attaches no value to any item.
With this dummy agent, there always exists a complete allocation, where all items are allocated.
Any complete EF allocation for this augmented agent set corresponds to an EF allocation in which the unallocated bundle is not envied by any (non-dummy) agent.

Since EF allocations barely exist in natural instances and are hard even to approximate, we are interested in the relaxed notion of {\em envy-freeness up to one item} (EF1) \cite{lipton2004approximately}, whereby envy is allowed but not for more than one item. 
It is shown in the aforementioned work of Wu et al. that any allocation which maximizes the Nash Social Welfare (NSW) is approximately EF1, with a tight approximation ratio of $1/4$.\footnote{This holds even in a more general setting in which an item may be of different values to different agents, as shown by Wu et al. 
%However, the approximation ratio of 1/4 is tight even if the valuations are identical.
}
%
%Their work leaves open a natural question:
An open question that follows this result is:

\begin{quote}
{\em Does there exist budget-feasible allocations with a better approximation ratio of EF1?}
%\jiarui{Does there exist any other class of budget-feasible allocations with a better approximation ratio (as approximately EF1 allocations)?}
\end{quote}

%In addition to their existence, it is also an interesting direction to study the computational complexity of such allocations.
%Observe that an NSW maximizing allocation per se is hard to compute even in the unconstrained setting~\cite{nguyen2014computational,lee2017apx}.
%%
%In other words, it remains unknown whether there exists polynomial time algorithm that computes (approximately) EF1 allocations that are budget-feasible.
%Consequently, we also consider the following open problem:

In addition, since an NSW maximizing allocation per se is hard to compute even in the unconstrained setting~\cite{nguyen2014computational,lee2017apx}, it remains unknown whether there exists an efficient algorithm that computes any approximately EF1 allocation.
Consequently, we are also interested in the following quesiton:

\begin{quote}
{\em Can we compute such allocations in polynomial time?}
%\jiarui{Can we design a polynomial-time algorithm to compute budget-feasible allocations that are approximately EF1, with a guaranteed approximation ratio?}
\end{quote}

As mentioned in \cite{corr/WuLG20}, the above questions appear to be non-trivial. 
In this paper, we answer both questions affirmatively.

\subsection{Main Results and Techniques}

%We show that for EF1 allocations, commonly used allocation algorithms can perform arbitrarily bad for our problem, even if they are modified by considering the values or densities (the ratio between value and size) of the items.
%One difficulty is that we cannot easily exchange the agents' bundles as we can in the unconstrained setting because they might have different budgets.
%We therefore introduce the concept of {\em virtual budgets}.
%We gradually increase the virtual budgets of agents in our algorithm to ensure that the budget constraints are not violated. Meanwhile, we are able to exchange bundles between agents who have the same virtual budget. Our algorithm runs in polynomial time and as one of our main contributions, we prove that it always generates a 1/2-approximate EF1 allocation.

We observe that with budget constraints, commonly used algorithms for computing EF1 allocations may perform arbitrarily bad, even with trivial modifications that takes into account  the values or densities (the ratio between the value and size) of the items.
One difficulty is that the knapsack problem (i.e., maximizing total value subject to budget constraint) does not have a succinct optimal solution and a small difference in budgets can make the optimal solutions significantly different. 
We therefore introduce the concepts of {\em feasible configuration} and  {\em virtual budget}.
Informally, a feasible configuration is a collection of item bundles that fit agents' budgets.
Our algorithm greedily assigns items with the largest density while ensuring that the resulting allocation remains a feasible configuration.
In addition, we gradually increase the virtual budgets and ensure that the virtual budget of an agent is at most her actual budget.
The virtual budgets enable us to exchange bundles between agents even if their actual budgets are different.
Our algorithm runs in polynomial time and generates a 1/2-approximate EF1 allocation.

Besides the approximability, we show that two special, yet typical, settings of our problem admit polynomial-time algorithms for computing an exact EF1 allocation.
In the setting where the agents have identical budgets, we show that our algorithm for computing a 1/2-approximate EF1 allocation in the general setting actually degenerates to an algorithm with an ``early stop'' feature, and it computes an exact EF1 allocation.
In the setting with two agents, our algorithm is based on the ``divide and choose'' approach, that is widely used in the fair division literature.

Finally, we consider the large budget setting, in which the agents' budgets are much larger than the sizes of items~\cite{esa/FeldmanHKMS10,ec/DevanurJSW11,icalp/MolinaroR12,stoc/KesselheimTRV14,corr/WuLG20}.
We show that under this setting, our polynomial-time algorithm computes an allocation whose approximation ratio of EF1 is close to $1$.
We also investigate the extent to which an NSW maximizing allocation approximates EF1 in the large budget setting. 
It has been shown that (for non-identical valuation functions) such an allocation is 1/2-approximate EF1~\cite{corr/WuLG20}.
We prove that when agents have identical valuations, this approximation ratio improves and approaches 1, which coincides with the result of Caragiannis et al.~\cite{caragiannis2019unreasonable} for the unconstrained setting.

\subsection{Related Work}

There is a growing research interest in fair division, especially since the relaxed notion of EF1 was introduced \cite{lipton2004approximately,budish2011combinatorial}. 
Our problem falls into the category of {\em constrained} fair division of {\em indivisible} items.
Existing works mainly consider the problem under matroid~\cite{conf/ijcai/BiswasB18,conf/aaai/BiswasB19} or cardinality constraints~\cite{corr/AzizHMS19,corr/HummelH21}.
While an EF1 allocation can be simply found by the round-robin algorithm in the unconstrained setting \cite{caragiannis2019unreasonable}, 
the problem becomes much more difficult when various constraints are taken into consideration. 
For example, with matroid constraints, it is shown that an EF1 allocation exists when agents have an identical valuation \cite{conf/ijcai/BiswasB18}. 
Later on, Biswas and Barman~\cite{conf/aaai/BiswasB19} provided a polynomial time algorithm to compute an EF1 allocation.
More recently, progress was also made on the setting with non-identical valuations~\cite{conf/aaai/DrorFS21,corr/WuLG20}.

The problem we consider in this paper can be viewed as a multi-agent version of the multiple knapsack problem (MKP).
The problem is a natural generalization of the classical NP-complete problem, the knapsack problem \cite{karp1972reducibility}.
In the MKP, a set of of items with varied sizes and values are to be packed in multiple knapsacks. Each knapsack has a budget that limits the total size of items it can take.
The goal is to find a way to pack the items so that the total value of packed items is maximized~\cite{ibarra1975fast,karmarkar1982efficient,martello1990knapsack,chekuri2005polynomial,kellerer2013knapsack}.
In other words, existing works study the MKP with respect to social welfare, while our work focuses on fairness.

The notions of fairness and budget/knapsack constraints also appear in voting scenarios \cite{conf/wine/FainGM16,conf/sigecom/ConitzerF017,conf/aaai/FluschnikSTW19}, where a set of voters vote for a set of costly items and the goal is to select a set of items within a fixed budget. Nevertheless, these problems are fundamentally different from ours.
There is only a single bundle to be selected and it is to be accessed by all the agents, whereas in our problem, we select a bundle for each agent, who does not have access to the other agents' bundles.

\section{Preliminaries}

In the budget-feasible fair allocation problem, a set $M$ of $m$ goods needs to be allocated to a set $N$ of $n$ agents.
%; each agent has a budget constraint on the items she receives.
Every item $j \in M$ has a value $v_j$ and a size $s_j$, and each agent $i \in N$ has a budget $B_i$, which restricts the total size of items she can receive.
We let $\bfs = (s_1, \dots, s_m)$, $\bfv = (v_1, \dots, v_m)$, and $\bfB = (B_1, \dots, B_n)$ be the size, value, and budget profiles, respectively; and we denote by $I = (\bfs, \bfv; \bfB)$ an instance of the problem.
For any subset $X\subseteq M$ of items, we let $s(X) = \sum_{j\in X}s_j$ and $v(X) = \sum_{j\in X}v_{j}$ be the total size and value of items in $X$, respectively.
We also denote by $\rho_{j} = v_{j}/s_j$ the {\em density} of each item $j\in M$, and by $\rho(X) = v(X)/s(X)$ the average density of items in a set $X\subseteq M$.
For notational simplicity, for any $X\subseteq M$ and $g\in M$, we write $X\cup\{g\}$ as $X+g$ and  $X\setminus\{g\}$ as $X - g$.

An allocation is an ordered $(n+1)$-partition of $M$, denoted as $\mathbf{X}=(X_0, X_1, \cdots, X_n)$, where each $X_i$ is the bundle of items allocated to agent $i \in N$ and $X_0$ contains the unallocated items.
The allocation must satisfy the budget constraints.
An allocation $\mathbf{X}$ is {\em budget-feasible} if $s(X_i) \le B_i$ for all $i \in N$.
For example, $(M, \emptyset, \cdots, \emptyset)$ is trivially a feasible allocation, in which every agent gets an empty bundle. 
Given an allocation $\mathbf{X}$, we say that agent $i$ is \emph{tight} if her remaining budget is insufficient for taking any unallocated item, i.e., $s(X_i+g) > B_i$ for all $g \in X_0$. We will also think of the unallocated items $X_0$ as endowment to a {\em charity} --- a special agent with an unlimited budget.
We let $N^+ = N \cup \{0\}$ be the set of agents including the charity.
Introducing the charity is important for our problem because the budget constraints disallow us to always allocate all the items to agents in $N$. 

We adapt the EF and EF1 notions to the above budget-feasible setting as follows.

\begin{definition}[$\alpha$-EF]
	For $0\le \alpha \le 1$, an allocation $\mathbf{X}$ is called {\em $\alpha$-approximate envy-free} or {\em $\alpha$-EF } 
	if for every pair of agents $i \in N$, $j \in N^+ - i$ and every $T\subseteq X_j$ with $s(T) \le B_i$,
	\begin{equation*}
		v(X_i) \ge \alpha \cdot v(T). 
	\end{equation*}
	When $\alpha = 1$, the allocation is also said to be EF.
\end{definition}

In other words, in an EF allocation, no agent $i$ finds a sub-bundle of another agent more valuable than her own bundle while this sub-bundle also fits with her budget.
The requirement that the agents do not envy the charity also excludes the allocation $(M, \emptyset, \cdots, \emptyset)$ from our consideration.
An $\alpha$-EF allocation may not exist for any $\alpha > 0$ even in the classical setting without budget constraints (e.g., when there are two agents but only one item to be allocated), so the next hope would be to find an EF1 allocation, which allows an agent to envy another but for at most one item. We define $\alpha$-EF1 below.

\begin{definition}[$\alpha$-EF1] \label{def:efk}
	For $0\le \alpha \le 1$, an allocation $\mathbf{X}$ is called {\em $\alpha$-approximate envy-free up to one item} or {\em $\alpha$-EF1}, if for every pair of agents $i \in N$, $j \in N^+ - i$, and every $T\subseteq X_j$ with $s(T) \le B_i$, there exists $e\in T$ such that
	\begin{equation*}
		v(X_i) \ge \alpha \cdot v(T- e).
	\end{equation*}
	When $\alpha = 1$, the allocation is also said to be EF1.
\end{definition}

\section{Warm Up} \label{sec:warm-up}

When there are no budget constraints, it is well-known that {\em round robin} is a simple and efficient procedure that always yields an EF1 allocation. 
In this procedure, the agents take turns to select a most valuable item from the unallocated items, until all items are selected.
It would then be tempting to think that this simple procedure can be easily adapted to our setting. 
There are two immediate difficulties.

First, the size of a selected item may exceed an agent's remaining budget, so this item cannot be allocated to this agent. A straightforward workaround is to restrict each agent's selection to items smaller than their remaining budget.
However, consider the instance shown in Table~\ref{tbl:hard-instance-1}.
A round-robin approach will first allocate items 1 and 2 to agents 1 and 2, respectively. 
After that, agent 1 becomes tight. If we stop the allocation procedure at this point, both agents would envy the charity for more than one item, whereas if we continue the procedure with agent 2 (who is not yet tight), agent 2 would get many more items than does agent 1, resulting in agent 1 envying her for more than one item.

\begin{table}[h]
	\begin{center}
		\begin{tabular}{c|cccc}
			\hline
			item & 1 & 2 & \dots & 100\\
			\hline
			value & $1$ & $0.5$ & $\dots$ & $0.5$ \\
			size & $1$ & $0.1$ & $\dots$ & $0.1$ \\
			\hline
		\end{tabular}
	\end{center}
	\caption{\label{tbl:hard-instance-1} A hard instance. There are two agents with budgets $B_1 = B_2 = 1$. Items 2 to 100 are identical.}
\end{table}%

It is not hard to see that the failure of the above attempt is mainly due to the inappropriate allocation of item 1, which has a big value but a small density.
This brings us to a more sophisticated greedy algorithm which takes into account densities of the items. 
We repeat the following steps until all the agents' bundles are finalized:
\begin{itemize}
	\item[1.]
	Pick an agent $i$ whose bundle has the lowest value among bundles that are not yet finalized.
	
	\item[2.]
	Pick the {\em densest} unallocated item that does not exceed the remaining budget of agent $i$ and allocate this item to $i$.
	If no such an item exists, finalize the bundle of agent $i$ and continue with the remaining agents.
\end{itemize}

The above algorithm tries to balance between values and densities. 
%it uses density as the criterion for selecting an item in Step 2, while the absolute values are considered in Step 1 through prioritizing the agent with the least valuable bundle so far.
Unfortunately, it does not get us any closer to EF1.
On the instance shown in Table~\ref{tbl:hard-instance-2}, it fails to generate an $\alpha$-EF1 allocation for any $\alpha > 0$.
Indeed, it generates an allocation with $X_1=\{1,3\}$ and $X_2=\{2\}$, whereby we have
\[
v(X_2) = 2\epsilon \le \frac{2\epsilon} {1-\epsilon} \cdot v(X_1 - e) 
\]
for any $e\in X_1$. The coefficient ${2\epsilon} /(1-\epsilon)$ can be made arbitrarily small by choosing an $\epsilon$ sufficiently close to $0$.

\begin{table}[h]
	\begin{center}
		\begin{tabular}{c|ccc}
			\hline
			items & 1 & 2 & 3 \\
			\hline
			values & 1 & $2\epsilon$ & $1-\epsilon$ \\
			sizes & $\epsilon$ & $2\epsilon$ & $1-\epsilon$ \\
			\hline
		\end{tabular}
	\end{center}
	\caption{\label{tbl:hard-instance-2} There are two agents with budgets $B_1 = B_2 = 1$.}
\end{table}%

\section{Computing a 1/2-EF1 Allocation}

\label{sec:1/2-EF1}

In this section, we present an efficient algorithm for computing a $1/2$-EF1 allocation.
Without loss of generality, we assume that the agents are indexed in ascending order of their budgets, i.e., $B_1\leq B_2\leq \ldots\leq B_n$.
We also assume that there is always a dummy item of value $0$ and size larger than $B_n$ in $M$, so that there is always at least one unallocated item in $X_0$.
The pseudo-code of our main algorithm is presented in Algorithm~\ref{alg:compute-EF1}.

\begin{algorithm}[t]
	\small
	\caption{Computing a $1/2$-EF1 allocation.\label{alg:compute-EF1}}
	{\bf Input:} An instance $I = (\bfs, \bfv; \bfB)$\\
	
	\smallskip
	
	Initialize $X_i \leftarrow \emptyset$ for each $i\in N$, and $X_0 \leftarrow M$;\\
	Initialize level $l(i) \leftarrow 1$ for each $i\in N$;\label{ln:set-level}\\
	Initialize active agents: $\mathcal{A}\leftarrow N$; %\tcp{active agents}
	
	\smallskip
	
	\While{$\mathcal{A} \neq \emptyset$\label{ln:while-outer}}{
		Pick an arbitrary agent $i\in \mathcal{A}$ with the minimum $v(X_i)$; \label{ln:pick-min-Xi} \\ %\footnotemark\\
		$U\leftarrow X_0$;\\
		\While{$U \neq \emptyset$\label{ln:while-inner}}
		{
			Pick an arbitrary item $g\in U$ with maximum density;\\
			\uIf{$\mathsf{TryFit}(X_i,g)$ is successuful \label{ln:call-tryfit}}
			{
				Commit the operations in $\mathsf{TryFit}(X_i,g)$ and
				break out of this while loop (Lines~\ref{ln:while-inner}--\ref{ln:while-inner-end});\label{ln:commit-tryfit}\\
			}
			\Else
			{
				$U \leftarrow U - g$;\label{ln:while-inner-end}\\
			}
		}
		
		\If{$U = \emptyset$}
		{
			Let $j$ be the largest index such that $l(j) = l(i)$;\\
			Swap $i$ and $j$, and items in $X_i$ and $X_j$;\\
			%$\mathcal{A} \leftarrow \{j \in \mathcal{A} : l(j) > l(i)\}$;\label{ln:finalize}\\
			$\mathcal{A} \leftarrow \{i+1, i+2, \dots, n\}$;\label{ln:finalize}\\
		}
		
		\label{ln:while-outer-end}
	}
	
	\smallskip
	
	{\bf Output:} allocation $\mathbf{X} = (X_1, \dots, X_n)$.
\end{algorithm}

\begin{algorithm}[t]
	\caption{$\mathsf{TryFit}(X_i, g)$ \label{alg:try-fit}} \small	
	
	$\ell \leftarrow i$; \\
	\smallskip
	
	\While{$s(X_\ell + g) > B_{l(\ell)}$\label{ln:while-size-larger}}
	{
		Let $j$ be the largest index such that $l(j) = l(\ell)$;\\
		\uIf{$j \neq \ell$}
		{
			Swap items in $X_\ell$ and $X_j$\label{ln:tryfit-swap};\\
			$\ell \leftarrow j$;\label{ln:tryfit:i-j}
		}
		\uElseIf{$l(\ell) < \ell$\label{ln:li-lt-i}}
		{
			$l(\ell) \leftarrow l(\ell)+1$\label{ln:tryfit-levelup};
		}
		\Else{\textbf{Output:} {
				\em unsuccessful} (and terminate);
		}
	}
	
	$X_\ell \leftarrow X_\ell + g$; \label{ln:add-g}\\
	
	\textbf{Output:} {\em successful}.\\
\end{algorithm}

\subsection{Virtual budget and $\mathsf{TryFit}$ Subroutine}

At a high level, Algorithm~\ref{alg:compute-EF1} uses a similar greedy approach:
it repeatedly picks an agent with the least valuable bundle among all bundles not yet finalized (i.e., among active agents in $\mathcal{A}$), and attempts to allocate a densest item to this agent.
The algorithm sets a \emph{virtual budget} for each agent, which is maintained through a level value $l(i)$ for each $i$; the virtual budget of agent $i$ is then $B_{l(i)}$; 
Initially, the virtual budget is $B_1$ for all the agents (Line~\ref{ln:set-level}).
The algorithm maintains the following two invariants throughout.

\begin{invariant}[Level invariant]
	For all $i,j \in N$, if $i < j$, then $l(i)\leq l(j) \le j$. %{\color{red} I added $\le j$}
\end{invariant}

\begin{invariant}[Size invariant]
	For all $i\in N$, $s(X_i) \in \left( B_{l(i)-1},B_{l(i)} \right]$ (where we let $B_0 = -\infty$).
\end{invariant}

At Line~\ref{ln:call-tryfit}, Algorithm~\ref{alg:compute-EF1} calls a subroutine $\mathsf{TryFit}$. 
The goal of this subroutine is to check whether, after item $g$ is added to $X_i$, there is a way to reallocate the bundles $X_1,\dots, X_n$ so that the budget constraints are satisfied.
We call a collection of bundles that admits such a reallocation a {\em feasible configuration}; see Definition~\ref{def:fsb-conf}.
In $\mathsf{TryFit}$, the size of the new bundle $X_i$ is compared with the virtual budget of agent $i$. If the virtual budget is insufficient for packing $X_i$, the subroutine tries to adjust the agent's virtual budget to the next level.
The fulfillment of the goal of $\mathsf{TryFit}$ relies on the two invariants defined above, which are maintained by $\mathsf{TryFit}$ in the meantime; see Lemmas~\ref{lmm:invariants} and \ref{lemma:feasible-config-tryfit}.
If $\mathsf{TryFit}(X_i, g)$ is unsuccessful for all unallocated items $g \in X_0$, then we finalize bundle $X_i$ and mark agent $i$ as inactive. 
We say that this agent $i$ is {\em actively finalized}. In the meantime we also finalize and deactivate all agents $j$ with $l(j) \le l(i)$. 
We say that each of these agents $j$ is {\em (passively) finalized}.

\begin{definition}[Feasible configuration]
	\label{def:fsb-conf}
	A collection of bundles $\mathbf{Y} = \{Y_1,\ldots,Y_n\}$ is called {\em a feasible configuration} if there exists a permutation $\sigma:N \to N$, such that $\left(Y_0, Y_{\sigma(1)}, Y_{\sigma(2)}, \dots, Y_{\sigma(n)}\right)$ is a feasible allocation,
	where $Y_0 = M \setminus \bigcup_{i=1}^n \overline{Y}_i$.
\end{definition}

\begin{lemma}
	\label{lmm:invariants}
	The level and size invariants are maintained after the operations in $\mathsf{TryFit}$ are committed (Line~\ref{ln:commit-tryfit}, Algorithm~\ref{alg:compute-EF1}).
\end{lemma}
\begin{proof}
	First, the level invariant holds before $\mathsf{TryFit}$ is called the first time, when we have $l(j) = 1 \le j$ for all $j\in N$. Since we only increase the level of the agent with largest index in this level and $l(j) \le j$ due to the condition at Line~\ref{ln:li-lt-i}, the level invariant is maintained throughout.
	
	The size invariant holds before $\mathsf{TryFit}$ is called the first time, when we have $X_i = \emptyset$ for all $i\in N$.
	Suppose that the size invariant is maintained at some point before we execute the operations of $\mathsf{TryFit}$; we argue that it is also maintained afterwards.
	Indeed, throughout the while loop in $\mathsf{TryFit}$, $X_\ell$ is the only bundle that may violate the size invariant as the levels of all other bundles do not change. 
	Since $l(\ell)$ increases by at most $1$ in each iteration of the while loop,
	when the while loop breaks out, we have $B_{l(\ell)-1} < s(X_\ell + g) \le B_{l(\ell)}$.
	The addition of $g$ into $X_i$ at Line~\ref{ln:add-g} then gives $B_{l(\ell)-1} < s(X_\ell) \le B_{l(\ell)}$, so the size invariant is also maintained for bundle $X_\ell$.
\end{proof}

\begin{lemma}
	\label{lemma:feasible-config-tryfit}
	$\mathsf{TryFit}(X_i, g)$ is successful if and only if $\{X_1,\ldots, X_i+g, \dots,X_n\}$ is a feasible configuration.
\end{lemma}
\begin{proof}
	Suppose that $\mathsf{TryFit}(X_i, g)$ is successful and, for each $j\in N$, let $X'_j$ denote the value of $X_i$ after the execution of operations in $\mathsf{TryFit}$.
	By Lemma~\ref{lmm:invariants}, $s(X'_j) \le B_{l(j)} \le B_j$ for each $j \in N$, so $(X'_1, \dots, X'_n)$ is a feasible allocation. 
	Indeed, $\mathsf{TryFit}$ only permutes the order of the bundles $X_1,\ldots, X_i+g, \dots,X_n$, so the collection of these bundles is a feasible configuration.	
	Conversely, suppose that the configuration $\{X_1,\ldots, X_i+g, \dots,X_n\}$ is feasible. We show that $\mathsf{TryFit}(X_i, g)$ must be successful.
	Note that $\mathsf{TryFit}(X_i, g)$ is unsuccessful only if the following conditions hold at some iteration of the while loop:
	\begin{itemize}
		\item $s(X_\ell + g) > B_\ell$;
		\item $\ell$ is the agent with the largest index at level $\ell$; and 
		\item $l(\ell) = \ell$.
	\end{itemize}
	Nevertheless, if these conditions hold, then the number of bundles in $\{X_1,\ldots, X_i+g, \dots,X_n\}$ of sizes larger than $B_\ell$ is at least $n-\ell+1$ given that: $s(X_\ell + g) > B_\ell$, and $l(j) \ge l(\ell)$ for all $j > \ell$ by the level invariant, which implies $s(X_j) > B_\ell$ by the size invariant.
	This requires $n-\ell+1$ agents with budget at least $B_\ell$ while we only have $n-\ell$ such agents, so the configuration is not a feasible one, which contradicts the assumption.	
\end{proof}

\subsection{Approximation Ratio Analysis}

Now we are ready to show the following key theorem about the approximation guarantee by Algorithm~\ref{alg:compute-EF1}.

\begin{theorem} \label{thm:compute-EF1}
	Algorithm~\ref{alg:compute-EF1} computes a 1/2-EF1 allocation in polynomial time.
\end{theorem}

Indeed, by the level and size invariants, we have $s(X_j) \le B_{l(j)} \le B_j$, so Algorithm~\ref{alg:compute-EF1} always computes a feasible allocation.
In addition, it is not hard to see that Algorithm~\ref{alg:compute-EF1} finishes in polynomial time as in every iteration of the while loop, either one item is removed from $X_0$ or at least one agent is removed from the active set $\mathcal{A}$.
In the sequel, we present the following two parts to complete the proof of Theorem~\ref{thm:compute-EF1}:
(i) 1/2-EF1 is guaranteed between the agents (Lemma~\ref{lmm:ef1-agents}); and
(ii) no agent envies the charity by more than one item (Lemma~\ref{lmm:ef1-agent-charity}). 
%(iii) Algorithm~\ref{alg:compute-EF1}  runs in polynomial time.
For ease of description, in what follows we denote by $(\overline{X}_1, \dots, \overline{X}_n)$ the output of Algorithm~\ref{alg:compute-EF1} (and $\overline{X}_0 = M \setminus \bigcup_{i=1}^n \overline{X}_i$) to distinguish with the bundles $X_0, X_1, \dots, X_n$ not yet finalized during the execution of the algorithm.
%{\color{red} Maybe it's easier if we use $(\overline{X}_1, \dots, \overline{X}_n)$ to denote the output of Algorithm~\ref{alg:compute-EF1}... It is quite confusing in the proofs below when we also talk about $X_i$ during the algorithm... I rewrote Lemma~\ref{lmm:ef1-agent-charity} in this way (the previous version is in the comment)}

\begin{lemma}
	\label{lmm:ef1-active-finalize}
	Suppose that agent $k$ is finalized in some iteration where agent $i$ is actively finalized.
	Then $B_k \le B_i$ and $v(\overline{X}_k) \ge v(\overline{X}_i)$.
\end{lemma}
\begin{proof}
	We have $v(\overline{X}_k) = v(X_k)$ and $v(\overline{X}_i) = v(X_i)$, where $X_k$ and $X_i$ denote the bundles in the iteration where $\overline{X}_k$ and $\overline{X}_i$ are finalized.
	According to Line~\ref{ln:pick-min-Xi}, $X_i$ has the minimum value among all active bundles, so we have $v(X_k) \ge v(X_i)$ and hence, $v(\overline{X}_k) \ge v(\overline{X}_i)$. In addition, according to Line~\ref{ln:finalize}, all agents $\ell > i$ remain active, so we have $k \le i$ and hence, $B_k \le B_i$.
\end{proof}

\begin{lemma}
	\label{lmm:ef1-agents}
	For any pair of agents $i,j\in N$ and any $T\subseteq \overline{X}_j$ such that $s(T)\leq B_i$, there exists an item $g\in T$ such that $v(\overline{X}_i)\geq \frac{1}{2} \cdot v(T-g)$.
\end{lemma}

\begin{proof}
	We first show that the following claim holds throughout Algorithm~\ref{alg:compute-EF1}.
	
	\begin{claim}
		\label{clm:ef1-agents-A}
		For any $i,j\in \mathcal{A}$ it holds that $v(X_i) \ge v(X_j - g)$ for some $g\in X_j$. 
	\end{claim}
	\begin{proof}
		Loosely speaking, the active bundles are EF1 among the active agents.
		Indeed, this claim holds trivially when the while loop begins, where we have $X_i = \emptyset$ for all $i\in \mathcal{A}$.
		It then suffices to show that the claim holds at the end of an iteration as long as it holds at the beginning of this iteration.
		Suppose that Claim~\ref{clm:ef1-agents-A} holds at the beginning of an iteration.
		If $U=\emptyset$ in this iteration, then no item is allocated while $\mathcal{A}$ is updated to a subset of itself at the end of this iteration, so Claim~\ref{clm:ef1-agents-A} holds.
		If $U\neq \emptyset$, then an item $g$ is allocated to bundle $X_i$ by the committed operations of $\mathsf{TryFit}$.
		Suppose that the $\mathsf{TryFit}$ results in the index of each bundle $j$ being changed to $\sigma(j)$. We denote the bundle right after the execution of $\mathsf{TryFit}$ as $X'_j$.
		Namely, we have $X_j = X'_{\sigma(j)}$ for all $j\neq i$, and $X_i + g = X'_{\sigma(i)}$.
		%By assumption, for every $j, \ell \in \mathcal{A}$, we have $v(X_i) \ge v(X_j - e)$ for some $e \in X_j$. Then if $i \notin \{j, \ell\}$, we have $$v(X'_i) \ge v(X_j - g) $$
		For every $j_1, j_2 \in \mathcal{A}$, $j_1 \neq j_2$: 
		\begin{itemize}
			\item
			If $j_2 \neq \sigma(i)$, we have 
			$$v( X'_{j_1}) \ge v(X_{\sigma^{-1}(j_1)}) \ge v(X_{\sigma^{-1}(j_2)}-e) = v(X'_{j_2} - e)$$
			for some $e \in X_{\sigma^{-1}(j_2)} = X_{j_2}$, where the second transition holds according to our assumption that Claim~\ref{clm:ef1-agents-A} holds when this iteration begins. 
			
			\item
			If $j_2 = \sigma(i)$, then $j_1 \neq \sigma(i)$ and we have 
			$$v( X'_{j_1}) = v(X_{\sigma^{-1}(j_1)}) \ge v(X_{i}) = v(X'_{j_2} - g),$$
			where $v(X_{\sigma^{-1}(j)}) \ge v(X_{i})$ because $X_i$ has the minimum value according to Line~\ref{ln:pick-min-Xi}.
		\end{itemize}
		
		Thus, Claim~\ref{clm:ef1-agents-A} holds for the new bundles $X'_j$ at the end of the iteration.
	\end{proof}

	The claim immediately implies that if bundle $\overline{X}_j$ is finalized before or in the same iteration with bundle $\overline{X}_i$, then it holds that
	$$v(\overline{X}_i) \ge v(X_{i'}) \ge v(X_j - g) = v(\overline{X}_j - g)$$
	for some $g\in X_j$, where we look at the moment right before $\overline{X}_j$ is finalized at Line~\ref{ln:finalize}, and $i'$ is the index of bundle $\overline{X}_i$ in that iteration.
	Thus, for any $T\subseteq \overline{X}_j$, we have  
	$$v(\overline{X}_i) \ge v(\overline{X}_j - g) \ge v(T - g) \ge \frac{1}{2} \cdot v(T-g),$$
	as desired. % {\color{red} Ask!}
	
	On the other hand, if $\overline{X}_j$ is finalized after $\overline{X}_i$, consider the moment right before $\overline{X}_i$ is finalized at Line~\ref{ln:finalize}.
	By Lemma~\ref{lmm:ef1-active-finalize}, it suffices to consider the situation when $\overline{X}_i$ is {\em actively} finalized.
	We have
	\begin{equation}
		\label{eq:ef1-agents-eq1}
		v(\overline{X}_i) = v(X_i) \ge v(X_{j'} - g)
	\end{equation}
	for some $g \in X_{j'} \subseteq \overline{X}_j$, where $j'$ is the index of bundle $\overline{X}_j$ at that moment.
	Observe that the way Algorithm~\ref{alg:compute-EF1} proceeds means that:
	\begin{itemize}
		\item[(i)]
		The items in $\overline{X}_j \setminus X_{j'}$, which are included into $\overline{X}_j$ in the subsequent iterations, cannot have a higher density than any item in $X_i$. In addition, since $\overline{X}_i$ is actively finalized, $s(X_i + e) > B_i$ for all $e \in \overline{X}_j \setminus X_{j'}$ as otherwise at least one more item could be included in $X_i$ and it would not be finalized in the iteration we consider.
		
		\item[(ii)]
		$s(X_{j'}) > B_i$, which is due to the size invariant and the fact that $l(j') > l(i)$ by Line~\ref{ln:finalize}.
	\end{itemize}
	
	Pick arbitrary $T \subseteq \overline{X}_j$ such that $s(T) \le B_i$.
	\begin{itemize}
		\item
		If $g\in T$, then we immediately have
		$$v(\overline{X}_i) \ge v(X_{j'} - g) \ge v(T - g) \ge \frac{1}{2} \cdot v(T - g),$$
		where we use the inequality $v(T) \le v(X_{j'})$ which is due to the fact that $X_{j'}$ contains the densest items in $\overline{X}_j$ and $s(X_{j'}) > B_i \ge s(T)$.
		
		\item
		If $g\notin T$, then let $T_1 = T \cap X_{j'}$ and $T_2 = T \setminus T_1$ (so $T_2 \subseteq \overline{X}_j \setminus X_{j'}$). 
		Using \eqref{eq:ef1-agents-eq1}, we have 
		$$v(\overline{X}_i) \ge v(X_{j'} - g) \ge v(T_1).$$ 
		In addition, since items in $X_i$ are as dense as any item in $T_2$ and $s(X_i + e) > B_i$ for all $e \in T_2$, we have
		\begin{align*}
			v(\overline{X}_i) = v(X_i) 
			& \ge s({X}_i)  \cdot \rho(T_2 - e)
			> (B_i - s(e) ) \cdot \rho(T_2 -e) \\
			& \ge s(T_2 - e) \cdot \rho(T_2 - e) = v(T_2 - e).
		\end{align*}
		Combining the above two inequalities gives the desired inequality as well:
		$$2 v(\overline{X}_i) \ge v(T_1) + v(T_2 - e) = v(T - e).$$
	\end{itemize}
	
	This completes the proof.
\end{proof}

To prove Lemma~\ref{lmm:ef1-agent-charity}, we use the following useful result.

\begin{lemma} \label{lem:greedy:charity}
	Suppose $X$ and $Y$ are two sets of items with $s(X) \le B$ and $s(Y) \le B$.
	For every item $g \in Y$, let $W_g = \{j \in X : \rho_j \ge \rho_g\}$ be the set of items in $X$ that are at least as dense as $g$.
	If $s(W_g + g) > B$ for all $g \in Y$, there exists an item $j \in Y$ such that $v(Y - j) \le v(X)$.
\end{lemma}
\begin{proof}
	Let $g^*$ be the densest item in $Y$, i.e., $\rho_{g^*} \ge \rho_g$ for all $g \in Y$.
	By assumption of the lemma we have $s(W_{g^*} + g^*) > B$, which implies
	\begin{equation*}
		s(Y - g^*) \leq B - s_{g^*} < s(W_{g^*}).
	\end{equation*}
	
	Note that $W_{g^*} \neq \emptyset$ since otherwise $s(W_{g^*} + g^*) = s_{g^*} \le s(Y) \le B$, which contradicts the assumption of this lemma.
	Since $g^*$ is the densest item in $Y$, by definition any item in $W_{g^*}$ is at least as dense as $g^*$ and any other items in $Y$:
	\begin{equation*}
		\rho(W_{g^*}) = \frac{v(W_{g^*})}{s(W_{g^*})} \ge \frac{v(Y - g^*)}{s(Y - g^*)} = \rho(Y-g^*).
	\end{equation*}
	
	It follows that 
	\begin{align*}
		v(X) \ge v(W_{g^*}) \ge s(W_{g^*}) \cdot \frac{v(Y - g^*)}{s(Y - g^*)} > v(Y - g^*). & \qedhere
	\end{align*}
\end{proof}

\begin{lemma}
	\label{lmm:ef1-agent-charity}
	For all $j\in N$ and $T\subseteq \overline{X}_0$ with $s(T)\leq B_j$, there exists $g\in T$ such that $v(\overline{X}_j)\geq v(T-g)$.
\end{lemma}
\begin{proof}
	Pick an arbitrary $e \in T$ and let $W_e = \{ e'\in \overline{X}_j : \rho_{e'} \geq \rho_e \}$ be the set of denser items in $\overline{X}_j$. We show that $s(W_e + e) > B_j$, so applying Lemma~\ref{lem:greedy:charity} concludes the proof.
	
	Suppose that at some iteration $\mathsf{TryFit}(X_{j'}, e)$ is called and $j'$ is the index of $\overline{X}_j$ in that iteration, i.e., $X_{j'} \subseteq \overline{X}_j$.
	Since $e$ is not added into $\overline{X}_j$, $\mathsf{TryFit}(X_{j'}, e)$ is not successful. 
	It must be that $s(X_{j'} + e) > B_j$: otherwise, $(\overline{X}_1, \dots, \overline{X}_{j-1}, {X_{j'} + e}, \overline{X}_{j+1}, \dots, \overline{X}_n)$ is a feasible allocation as the $i$-th bundle in this allocation has size at most $B_i$ for all $i\in N$, which means that $\{ X_1,\dots,X_{j-1},X_{j'}+e, X_{j+1}, \dots, X_n\}$ is a feasible configuration, contradicting the fact that $\mathsf{TryFit}(X_{j'}, e)$ is not successful according to Lemma~\ref{lemma:feasible-config-tryfit}.
	The way Algorithm~\ref{alg:compute-EF1} proceeds ensures that an item added earlier into $\overline{X}_j$ is at least as dense as the ones added later on. Hence, $\rho(e') \ge \rho(e)$ for all $e' \in X_{j'}$, which means $X_{j'} \subseteq W_e$ and $s(W_e + e) \ge s(X_{j'} + e) > B_j$. Thus, Lemma~\ref{lem:greedy:charity} implies that there exists $g\in T$ such that $v(\overline{X}_j)\geq v(T-g)$.
	
	It remains to consider the case when $\mathsf{TryFit}(X_{j'}, e)$ is not called throughout the algorithm. 
	Consider the iteration where $\overline{X}_j$ is finalized and according to Lemma~\ref{lmm:ef1-active-finalize}, it suffices to consider the situation where $\overline{X}_j$ is {\em actively} finalized.
	%Let $i$ be the agent with the minimum $v(X_i)$ that is selected at Line~\ref{ln:pick-min-Xi} and let $j^*$ be the agent with the largest index such that $l(j^*) = l(i)$.
	This means that $\mathsf{TryFit}(X_i, e)$ is not successful. 
	%Hence, in the $\mathsf{TryFit}$ procedure, $X_i$ will first be swapped with $X_{j^*}$, and by Lines~\ref{ln:li-lt-i} and \ref{ln:tryfit-levelup}, we will have $l(i) = j^*$ at some point while $s(X_i + e) > B_{j^*}$ (so the while loop does not break out). 
	Hence, $s(X_i + e) > B_{i}$ and the same argument above implies that, for any $T\subseteq \overline{X}_0$ such that $s(T)\leq B_{i}$, there exists $g\in T$ such that $v(\overline{X}_i)\geq v(T-g)$. This completes the proof.
\end{proof}

\subsection{Tightness of the Analysis}

Our analysis is tight due to the instance presented in Table~\ref{tbl:tight-approx}, for which Algorithm~\ref{alg:compute-EF1} computes an allocation whose approximation ratio of EF1 is at most $1/2$.
It can be verified that when $\epsilon$ is sufficiently small, running Algorithm~\ref{alg:compute-EF1} on this instance results in $X_1 = \{1,3\}$ and $X_2=\{2,4,5,6\}$.
%with the following allocation order: item 1 $\to X_1$; item 2 $\to X_2$; item 3 $\to X_1$; item 4 $\to X_2$ (and $l(2) \leftarrow2$); bundle $X_1$ finalized; items 5 and 6 $\to X_2$.
However, $s(\{2,5,6\}) = 2-\epsilon \le B_1$ while even after removing the most valuable item, the remaining value of the bundle approaches $2 \cdot v(X_1)$ when $\epsilon$ goes to $0$. 

\begin{table}[h]
	\begin{center}
		\begin{tabular}{c|cccccc}
			\hline
			items & 1 & 2 & 3 & 4 & 5 & 6  \\
			\hline
			values & $\epsilon$  & $1$ & $1$ & $2-\epsilon$ & $1-3\epsilon$ & $1-3\epsilon$ \\
			sizes & $\epsilon^3$ & $\epsilon$  & 1  & $2$  & $1-\epsilon$ & $1-\epsilon$   \\
			\hline
		\end{tabular}
	\end{center}
	\label{default}
	\caption{There are two agents with $B_1 = 2-\epsilon$ and $B_2 = 100$. The items are ordered in descending order of their densities.\label{tbl:tight-approx}}
\end{table}%

\section{EF1 Allocations in Two Special Settings} \label{sec:Special}

We show that we can efficiently compute an EF1 allocation in two special settings.

\subsection{Identical Budget}

The windfall of Algorithm~\ref{alg:compute-EF1} is that it generates an exact EF1 allocation when the agents have the same budget.
In fact, the algorithm degenerates to a greedy algorithm with an early stop  presented as Algorithm~\ref{alg:uniform-budget}:
it terminates immediately when an agent with the least valuable bundle gets tight.

\begin{algorithm}[h]
	\small
	\caption{EF1 allocation for uniform-budget.\label{alg:uniform-budget}}
	\text{{\bf Input:} An instance $I = (\bfs, \bfv; \mathbf{C})$ with $B_i = C$ for all $i\in N$.}\\
	
	\smallskip
	
	$X_i \leftarrow \emptyset$ for each $i\in N$, and $X_0 \leftarrow M$; %\\
	$\mathcal{A}\leftarrow N$; 
	
	\smallskip
	
	\While{$\mathcal{A} \neq \emptyset$}
	{
		\text{Pick an (arbitrary) agent $i\in \mathcal{A}$ with the minimum $v(X_i)$;}\\
		$U \leftarrow \{ g \in X_0 : s(X_i + g) \le C \}$;\\
		
		\uIf{$U \neq \emptyset$}
		{
			Allocate a (arbitrary) densest item $g^* \in U$ to $i$:
			$X_i \leftarrow X_i +g^*$ and $X_0 \leftarrow X_0 - g^*$\\
		}
		\lElse{go to output}		
	}
	
	\smallskip
	
	{\bf Output} allocation $\mathbf{X} = (X_1, \cdots, X_n)$.
\end{algorithm}

\begin{theorem} 
	\label{lemma:idenV+idenB}
	Algorithm~\ref{alg:uniform-budget} computes an EF1 allocation in polynomial time when all agents have the same budget.
\end{theorem}
\begin{proof}
	To see that the allocation $\mathbf{X}$ computed by the algorithm is EF1, we first argue that it is EF1 among the agents.
	Apparently, the agents do not envy each other at the beginning of the algorithm, when everyone gets an empty bundle. 
	As the algorithm proceeds, it allocates an item $g$ to the agent $i$ with minimum value $v(X_i)$, which is obviously not envied by any other agent.
	Hence as long as the EF1-ness holds in the previous iteration, it remains to hold after $g$ is allocated to $i$: removing $g$ from the $X_i$ will remove any possible envy against $i$. 
	Thus, by induction, $\mathbf{X}$ is EF1 among the agents.
	
	It remains to show that no agent envies the charity by more than one item, which we prove using Lemma~\ref{lem:greedy:charity}.
	Without loss of generality, suppose that when the algorithm terminates, agent $1$ gets the bundle with smallest value.
	It suffices to argue that agent $1$ does not envy charity by more than one item, since any other agent has the same budget as agent $1$, and have a bundle at least as valuable as $X_1$.
	
	Fix an arbitrary $T \subseteq X_0$ with $s(T) \le B$. 
	For any $g \in T$, consider $W_g = \{j \in X_1 : \rho_j \ge \rho_g\}$.
	If $s(W_g + g) \leq B$, then $g$ should have been added into $X_1$ before any item in $X_1 \setminus W_g$ were added in, by the greedy allocation of the algorithm. 
	Since $g$ is not allocated to $X_1$, we have $s(W_g + g) > B$.
	Thus, by Lemma \ref{lem:greedy:charity}, there exists an item $j \in T$ such that $v(X_1) \ge v(T-j)$.
	Consequently, agent $1$ does not envy the charity by more than one item.
	Since Algorithm~\ref{alg:uniform-budget} allocates one item in each iteration, and each iteration finishes in polynomial time, the polynomial runtime is readily seen.
\end{proof}

\subsection{Two Agents}

The scenario with two agents is a typical special setting of interest in the literature of fair division, where many fairness criteria can often be guaranteed, such as EF1, EFX (envy-free up to {\em any} item) \cite{caragiannis2019unreasonable,plaut2020almost}, and MMS (maximin share fairness) \cite{budish2011combinatorial}. 
These fair allocations, especially EFX and MMS allocations, 
are often computed by an approach known as ``divide and choose'', which however can take exponential time to finish. 
A natural question is to investigate whether this approach can be applied to our problem with budget constraints.
We present Algorithm~\ref{alg:two-agents}, a divide-and-choose style algorithm that computes an exact EF1 allocation.
Interestingly, our algorithm runs in polynomial time.
The correctness of the algorithm is shown in Theorem \ref{thm:n=2}, the proof of which can be found in the supplementary material.

\begin{algorithm}[h]
	\small
	\caption{EF1 allocation for two agents.\label{alg:two-agents}}
	
	\text{{\bf Input:} An instance $I = (\bfs, \bfv; \mathbf{C})$ with $n=2$ and $B_1 \le B_2$.}\\
	
	$\mathbf{C}' \leftarrow (B_1, B_1)$;\\
	
	$(X'_1, X'_2)\leftarrow$ output of Alg.~\ref{alg:uniform-budget} on instance $(\bfs, \bfv; \mathbf{C}')$;\\
	
	$X_1 \leftarrow \arg\min_{X \in \{X'_1, X'_2\}} v(X)$. If there is a tie (in value), then let $X_1$ be the bundle with smaller size;\\
	
	$\tilde{s}(g) \leftarrow \infty$ for all $g \in X_1$;\\
	\text{$X_2 \leftarrow$ output of Alg.~\ref{alg:uniform-budget} on single-agent instance $(\tilde{\bfs}, \bfv; B_2)$;}\\
	
	%	\smallskip
	
	{\bf Output} allocation $\mathbf{X} = (X_1, X_2)$.
\end{algorithm}

\begin{theorem}
	\label{thm:n=2}
	Algorithm~\ref{alg:two-agents} computes an EF1 allocation in polynomial time when there are 2 agents.
\end{theorem}
\begin{proof}
	In the following, we refer to Line 2 - 4 as Phase-(1) and Line 5 -7 as Phase-(2) of Algorithm~\ref{alg:two-agents}.
	
	Let $X_0 = M\setminus (X_1\cup X_2)$ be the unallocated items.
	We first show that agent $1$ does not envy agent $2$ or the charity by more than one item.
	Indeed, we prove the following stronger statement: for any $T\subseteq X_2\cup X_0 = M\setminus X'_1$ with $s(T) \leq B_1$, there exists an item $j\in T$ such that
	\begin{equation*}
		v(X_1) \geq v(T - j).
	\end{equation*}
	
	Obviously, if $T\subseteq X'_2$, then the statement is true because
	\begin{equation*}
		v(X_1) = v(X'_1) \geq v(X'_2) \geq v(T).
	\end{equation*}
	
	Otherwise let $j$ be the item in $T\cap X'_0$ (where $X'_0 = M\setminus (X'_1\cup X'_2)$) with the maximum density and consider the moment in Phase-(1) when the algorithm tries to include item $j$ into the bundle $X'_2$.
	Let $X''_2 \subseteq X'_2$ be the items that are already allocated to the bundle $X'_2$ at this moment.
	It is possible that $X''_2 \neq X'_2$ because some smaller items may get allocated to $X'_2$ after item $j$.
	Since $j\in X'_0$ at the end of Phase-(1), item $j$ is not successfully included into bundle $X''_2$ at this moment, i.e., we have that $s(X''_2 + j) > B_1$.
	Thus we have
	\begin{equation*}
		s(T-j) \leq B_1 - s_j \leq s(X''_2).
	\end{equation*}
	
	Moreover, each item in $X''_2$ has density at least $\rho_j$; each item in $(T-j)\setminus X''_2$ has density at most $\rho_j$.
	Consequently
	\begin{equation*}
		v(T - j) \leq v(X''_2) \leq v(X'_2) \leq v(X_1).
	\end{equation*}
	
	Next we show that agent $2$ does not envy agent $1$ or the charity by more than one item.
	It is easy to see that agent $2$ does not envy charity because the allocation $X_2$ is returned by Algorithm~\ref{alg:uniform-budget}, which guarantees EF1-ness between $X_2$ and $X_0$ (see Theorem~\ref{thm:compute-EF1}).
	To show that agent $2$ does not envy agent $1$, it suffices to prove that $v(X_2) \geq v(X'_2)$ because
	\begin{equation*}
		v(X'_2) \geq v(X_1 - j_\text{last}),
	\end{equation*}
	where $j_\text{last} \in X_1 = X'_1$ is the last item assigned to bundle $X'_1$ by Algorithm~\ref{alg:uniform-budget} in Phase-(1).
	
	Observe that $X'_2$ can be produced by running Algorithm~\ref{alg:uniform-budget} with a single agent with budget $B_1$ on items $M\setminus X'_1$.
	Let $j^*\in X_2\setminus X'_2$ be the first item that is included into $X_2$ that is not from $X'_2$.
	If no such item exists then we have $X_2 = X'_2$ (since $B_2 \geq B_1$) and the $v(X_2)\geq v(X'_2)$ trivially holds.
	Note that right before the algorithm tries to allocate item $j^*$ to $X_2$ and $X'_2$, both bundles contain exactly the same set of items $Y$.	
	Since $j^*$ is included in $X_2$ but not in $X'_2$, we have
	\begin{equation*}
		B_1 < s(Y+j^*) \leq B_2,
	\end{equation*}
	which implies that $s(Y+j^*) > s(X'_2)$.
	Moreover, since items in $X'_2\setminus Y$ have density at most that of $j^*$, we have $v(Y+j^*) > v(X'_2)$, which implies $v(X_2) \geq v(Y+j^*) > v(X'_2)$.
\end{proof}

\section{Large Budget Setting}
\label{sec:Nash}

In this section we consider the case when the budgets of agents are sufficiently large compared with sizes of items, i.e., \emph{large budget case}~\cite{esa/FeldmanHKMS10,ec/DevanurJSW11,icalp/MolinaroR12,stoc/KesselheimTRV14,corr/WuLG20}.
We show that under the large budget setting, our polynomial-time algorithm (Algorithm~\ref{alg:compute-EF1}) presented in Section~\ref{sec:1/2-EF1} computes an allocation that is almost EF1.
We also show that an allocation that maximizes NSW is almost EF1.
Let $\kappa = \min_{i\in N, j\in M} ({B_i}/{s_j})$, so every item has size at most ${B_i}/{\kappa}$, for any agent $i\in N$.
We consider the case when $\kappa$ is large in this section.

\subsection{Computation of Almost EF1 Allocation}

Recall that in Section~\ref{sec:1/2-EF1}, we propose a polynomial-time algorithm (Algorithm~\ref{alg:compute-EF1}) that computes an allocation that is $1/2$-EF1.
In the following, we show that when $s_j \leq B_i/\kappa$ for every item $j\in M$ and agent $i\in N$, the same algorithm computes an allocation that is $(1-1/\kappa)$-EF1.

\begin{theorem}\label{thm:compute-almost-EF1}
	Algorithm~\ref{alg:compute-EF1} computes a $(1-1/\kappa)$-EF1 allocation in polynomial time.
\end{theorem}
\begin{proof}
	By Lemma~\ref{lmm:ef1-agent-charity}, no agent envies the charity for more than one item. 
	Thus it suffices to show that for any two agents $i,j \in N$ and any $T\subseteq \overline{X}_j$ with $s(T)\leq B_i$, there exists an item $g\in T$ such that $v(\overline{X}_i)\geq (1-1/\kappa)\cdot v(T-g)$.
	Similar to the proof of Lemma~\ref{lmm:ef1-agents}, by Claim~\ref{clm:ef1-agents-A}, if bundle $\overline{X}_i$ is finalized after bundle $\overline{X}_j$, then agent $i$ does not envy agent $j$, and thus bundle $T$, for more than one item.
	Otherwise, i.e., $\overline{X}_j$ is finalized after $\overline{X}_i$, it suffices to consider the situation when $\overline{X}_i$ is actively finalized, by Lemma~\ref{lmm:ef1-active-finalize}.
	We show that in this case we have $v(\overline{X}_i) \geq (1-1/\kappa)\cdot v(T)$.
	
	Consider the moment right before $\overline{X}_i$ is finalized, and let $j'$ be the index of the bundle $\overline{X}_j$ at this moment.
	Let $g\in {X}_{j'}$ be the last item that is included in ${X}_{j'}$.
	We have $v(\overline{X}_i) \geq v({X}_{j'}-g)$.
	By the size invariant we have $s(X_{j'}) > B_i$, which implies that $v(X_{j'}-g) \geq (1-1/\kappa)\cdot v(X_{j'})$, because $g$ has the minimum density among items in $X_{j'}$, and $s_g \leq B_i/\kappa$.
	Finally, since both $X_{j'}$ and $T$ are subsets of $\overline{X}_j$, $s(X_{j'}) > B_i \geq s(T)$ and $X_{j'}$ contains the densest items of $\overline{X}_j$, we have $v(X_{j'}) > v(T)$.
	Putting everything together, we have
	\begin{equation*}
		v(\overline{X}_i) \geq v({X}_{j'}-g) \geq (1-1/\kappa)\cdot v(X_{j'})
		\geq (1-1/\kappa)\cdot v(T),
	\end{equation*}
	as we claimed.
\end{proof}

\subsection{Nash Social Welfare}

In fair division domain, one widely discussed topic is the fairness of the allocation that maximizes Nash Social Welfare (NSW), which is defined as the product of the agents' values, i.e., $\prod_{i \in N} v(X_i)$.
Under budget constraints, it is shown in \cite{corr/WuLG20} that such an allocation is 1/4-EF1.
They also show that even when the large budget setting, this allocation cannot guarantee an approximation better than 1/2. 
We improve this result by showing that when the agents have identical valuations for the items, the NSW maximizing allocation is almost EF1.
Note that when $m<n$ the optimal NSW is $0$, since some agents must have value $0$.
In this case, the NSW maximizing allocation allocates each item to a unique agent, which is trivially EF1.
Thus, the cases of our interest are the ones in which $m\geq n$.

\begin{theorem}\label{thm:Nash-almost-EF1}
	An allocation $\mathbf{X}^*$ that maximizes NSW is $(1-4\cdot \kappa^{-\frac{1}{4}})$-EF1.
\end{theorem}
\begin{proof}
	For ease of notation let $k = \kappa^{\frac{1}{4}}$.
	Thus every item has size at most ${B_i}/{k^4}$, for any agent $i\in N$.
	It suffices to consider the case when $k \geq 6$, as otherwise Theorem~\ref{thm:Nash-almost-EF1} is implied by the result of~\cite{corr/WuLG20}, who proved an $1/4$ approximation of EF1 for $\mathbf{X^*}$.
	Suppose $\mathbf{X}^*$ is not $(1-{4}/{k})$-EF1.
	Then there exists two agents, say agent $1$ and $2$, such that there is $T\subseteq X^*_2$ such that for all item $j\in T$,
	\begin{equation*}
		v(T-j) > \frac{1}{1-4/k}\cdot v(X^*_1)
		= (1 +\frac{4}{k-4})\cdot v(X^*_1).
	\end{equation*}
	Specifically, the inequality holds for the item $j^*\in T$ with the largest value $v_{j^*}$.
	Suppose $v(T-j^*) = (1+\epsilon)\cdot v(X^*_1)$, for some $\epsilon > \frac{4}{k-4}$.
	In the following, we consider two cases depending on the values of items in $T-j^*$.
	
	\paragraph{Case 1.}
	There exists $j\in T-j^*$ s.t. $v_j \geq \frac{1}{k^3}\cdot v(T-j^*)$.
	
	\smallskip
	
	Let $\Delta\subseteq X^*_1$ be a set of items with minimum value $v(\Delta)$ so that after removing $\Delta$ from $X^*_1$, 
	item $j$ can be included into $X^*_1$ without exceeding its budget, i.e., $s(X^*_1\setminus \Delta) \leq B_1-s_j$.
	
	\begin{claim}
		\label{claim:minimal-j}
		$v(X^*_1\setminus \Delta) \geq (1 - \frac{2}{k^4})\cdot v(X^*_1)$.
	\end{claim}	
	\begin{proof}
		Since every item has size at most $\frac{B_1}{k^4}$, if $s(X^*_1) \le (1 - \frac{1}{k^4})B_1$, the above inequality trivially holds with $\Delta = \emptyset$.
		If $s(X^*_1) > (1 - \frac{1}{k^4})B_1$, we can construct a subset $\Delta'$ by taking items with minimum density from $X^*_1$ until $s(X^*_1 \setminus \Delta) \leq B_1 - s_j$. 
		Since each item has size at most $B_1/k^4$, we must have $s(X^*_1 \setminus \Delta) > (1-2/k^4)\cdot B_1$ when we stop.
		Consequently we have
		\begin{equation*}
			v(X^*_1\setminus \Delta') \ge \frac{s(X^*_1\setminus \Delta')}{s(X^*_1)}\cdot v(X^*_1) \ge  (1 - \frac{2}{k^4})\cdot v(X^*_1).
		\end{equation*}
		Since $\Delta$ is the set of item with minimum value satisfying $s(X^*_1\setminus \Delta) \leq B_1 - s_j$, we have $v(X^*_1\setminus\Delta) \ge v(X^*_1\setminus \Delta') \ge  (1 - {2}/{k^4})\cdot v(X^*_1)$.
	\end{proof}
	
	Next we define a new allocation $\mathbf{X}$.
	
	Let $j\in X_2$ and $\Delta\subseteq X^*_1$ be as specified above.
	Define $X_1 = (X^*_1\setminus \Delta)+j$ and $X_2 = X^*_2 - j$. 
	Note that $X_1$ and $X_2$ satisfy the budget constraints.
	Denote 
	\[
	\delta = \frac{v_j}{v(T-j^*)} = \frac{v_j}{ (1+\epsilon)\cdot v(X^*_1)}\geq \frac{1}{k^3},
	\]
	and we have
	\begin{equation*}
		v(X_1) \geq \left( 1-\frac{2}{k^4}+ \delta(1+\epsilon) \right) v(X^*_1)
		> \left( 1 + \delta \right) v(X^*_1),
	\end{equation*}
	where the second inequality holds since 
	\begin{align*}
		\delta \cdot \epsilon - \frac{2}{k^4} > \frac{1}{k^3} \cdot \frac{4}{k - 4} - \frac{2}{k^4} > \frac{2}{k^4} > 0.
	\end{align*}
	On the other hand, since $T\subseteq X^*_2$ and $v_{j^*} \geq v_j$, we have
	\begin{align*}
		v(X_2) & \geq v(X^*_2) - \frac{v(X^*_2)}{v(T-j^*) + v_{j^*}}\cdot v_j \\
		& \geq v(X^*_2) - \frac{v_j}{v(T-j^*) + v_{j}}\cdot v(X^*_2) \\
		& = \left( 1 - \frac{\delta}{\delta+1} \right)\cdot v(X^*_2).
	\end{align*}
	
	Let $X_i = X^*_i$ for all $i \ge 3$ and $X_0 = X^*_0 \cup \Delta$.
	Then
	\begin{align*}
		\frac{\prod_{i\in N}v(X_i)}{\prod_{i\in N}v(X^*_i)} > \left( 1 + \delta \right)\cdot \left( 1 - \frac{\delta}{1+\delta} \right) =1,
	\end{align*}
	which contradicts the optimality of $\mathbf{X}^*$.
	
	\paragraph{Case 2.}
	For all $j\in T-j^*$, we have $v_j < \frac{1}{k^3}\cdot v(T-j^*)$.
	For any set $Y\subseteq T-j^*$, let $\Delta(Y) \subseteq X^*_1$ be the set of items in $X^*_1$ with $s(X^*_1\setminus \Delta(Y)) \leq B_1-s(Y)$ and minimum $v(\Delta(Y))$.
	That is, $\Delta(Y)$ is the bundle with smallest value such that after removing $\Delta(Y)$ from $X^*_1$, we can include $Y$ into $X^*_1$ without violating the budget constraint.
	
	Recall that $v(T-j^*)=(1+\epsilon)\cdot v(X^*_1)$.
	
	\begin{claim}\label{claim:minimal-Y}
		There exists $Y\subseteq T-j^*$ such that $Y \neq \emptyset$ and
		\begin{equation}
			v((X^*_1 \setminus \Delta(Y))\cup Y)\geq (1+\frac{\epsilon}{2})\cdot v(X^*_1).
			\label{eq:threshold-for-1}
		\end{equation}
		Moreover, for any $j\in Y$ we have
		\begin{equation*}
			v((X^*_1 \setminus \Delta(Y))\cup (Y-j)) < (1+\frac{\epsilon}{2})\cdot v(X^*_1).
		\end{equation*}
	\end{claim}
	\begin{proof}
		Observe that $Y \neq \emptyset$ does not satisfy Inequality~\eqref{eq:threshold-for-1}.
		On the other hand, for $Y = T-j^*$, we have
		\begin{equation*}
			v((X^*_1 \setminus \Delta(Y))\cup Y)\geq v(Y) = (1+\epsilon)\cdot v(X^*_1),
		\end{equation*}
		which satisfies Inequality~\eqref{eq:threshold-for-1}.
		
		Let $Y$ be the minimal subset of $T-j^*$ that satisfies~\eqref{eq:threshold-for-1}.
		Then $Y$ also satisfies the second condition.
	\end{proof}
	
	Next we define a new allocation $\mathbf{X}$.
	Let $Y\subseteq T-j^*$ be as specified in Claim~\ref{claim:minimal-Y}.
	Define $X_1 = (X^*_1\setminus \Delta(Y))\cup Y$ and we have 
	\[
	v(X_1) \geq (1+\frac{\epsilon}{2})\cdot v(X^*_1).
	\]
	
	Let $X' = (X^*_2\setminus Y)\cup \Delta(Y)$, which might have size larger than $B_2$.
	Note that by Claim~\ref{claim:minimal-Y}, for any item $j\in Y$, we have
	\begin{align}
		v(Y) - v(\Delta(Y)) & = v(Y-j)-v(\Delta(Y)) + v_j
		< \frac{\epsilon}{2}\cdot v(X^*_1) + \frac{1}{k^3}\cdot v(T-j^*) \nonumber \\
		& = \left( \frac{\epsilon}{2(1+\epsilon)} + \frac{1}{k^3} \right)\cdot v(T-j^*). \label{eq:Y-DY}
	\end{align}
	
	Since the size of $X'$ might be larger than $B_2$, we cannot assign all these items to agent $2$.
	However, observe that $s(\Delta(Y))\leq s(Y)+{B_2}/{k^4}$, as otherwise we can remove some item (which has size at most ${B_2}/{k^4}$) from $\Delta(Y)$, which contradicts the definition of $\Delta(Y)$.
	Hence we have $s(X')\leq B_2 + \frac{B_2}{k^4}$.
	Define $X_2$ to be the set obtained from $X'$ by removing items with minimum density until the size becomes at most $B_2$.
	Similar to Claim \ref{claim:minimal-j}, we have 
	\begin{align*}
		v(X_2) &\geq (1-\frac{2}{k^4})\cdot v(X')
		\geq (1-\frac{2}{k^4})\cdot v(X^*_2) - (v(Y) - v(\Delta(Y))) \\
		& \geq \left( 1 - \frac{2}{k^4} - \frac{\epsilon}{2(1+\epsilon)} - \frac{1}{k^3} \right)\cdot v(X^*_2)
		= \left( 1 - \frac{\epsilon}{2} + \frac{\epsilon^2}{2(1+\epsilon)} - \frac{2}{k^4} - \frac{1}{k^3} \right)\cdot v(X^*_2) \\
		& \geq \left( 1 - \frac{\epsilon}{2} + \frac{0.9\epsilon^2}{2(1+\epsilon)}\right)\cdot v(X^*_2),
	\end{align*}
	where the third inequality is by Equation (\ref{eq:Y-DY}) and the last inequality uses $\epsilon > \frac{4}{k-4}$ and $k\geq 6$.
	
	Finally, we have
	\begin{align*}
		\frac{\prod_{i\in N}v(X_i)}{\prod_{i\in N}v(X^*_i)} & > \left( 1 + \frac{\epsilon}{2} \right)\cdot \left( 1 - \frac{\epsilon}{2} + \frac{0.9\epsilon^2}{2(1+\epsilon)} \right) \\
		& = 1 - \frac{\epsilon^2}{4} + \frac{0.9\epsilon^2}{2(1+\epsilon)}(1+\frac{\epsilon}{2}) > 1,
	\end{align*}
	which is also a contradiction.
\end{proof}

For the setting without budget constraints, we have, equivalently, $B_i = \infty$ for all $i\in N$ and hence $\kappa = \infty$.
Thus, Theorem~\ref{thm:Nash-almost-EF1} implies that for the case without budget constraints, the allocation $\mathbf{X}^*$ is EF1, which coincides with the result of Caragiannis et al.~\cite{caragiannis2019unreasonable}.

\section{Conclusion and Future Directions}
In this work, we initiated the study of fair division under budget constraints.
We designed a novel polynomial-time algorithm that always returns a 1/2-EF1 allocation.
When agents have identical budget or there are only two agents, we show that an exact EF1 allocation can be computed in polynomial time.
A direct open problem is the existence of exact EF1 allocations for the general case or more general when agents have non-identical valuations.
There are also many other further directions. 
For example, in the current work, we restrict agents to not envy the charity to ensure efficiency.
It is interesting to consider alternative efficiency criteria such as Pareto optimality (PO). 
Finally, we believe it would be interesting to investigate the budget-feasible fair allocation problem under other fairness notions, e.g., maximin share fairness.

\bibliographystyle{abbrv}% the recommended bibstyle
\bibliography{fair-packing}

\end{document}